\documentclass[a4paper,10pt,oneside]{article}
\usepackage[utf8]{inputenc}
\usepackage{amsthm}
\usepackage{amsmath}
\usepackage{amscd}
\usepackage{amssymb}
\usepackage{algorithm}
\usepackage{algpseudocode}
\usepackage{mathtools}
\usepackage[T1]{fontenc}
\usepackage{authblk}
\usepackage{graphicx}
\usepackage{stmaryrd}
\usepackage[top=2cm,bottom=2.5cm,left=2.5cm,right=2.5cm,includefoot]{geometry}
\usepackage{multicol}
\usepackage{subcaption}
\usepackage{array}
\usepackage{pdfpages}

\usepackage{hyperref}
\usepackage{tikz}
\usetikzlibrary{shapes,arrows.meta,positioning,decorations.pathmorphing,decorations.pathreplacing,calc}

\hypersetup{
    colorlinks,
    citecolor=black,
    filecolor=black,
    linkcolor=black,
    urlcolor=black
}
\newcolumntype{"}{@{\hskip\tabcolsep\vrule width 1.5pt\hskip\tabcolsep}}

\newtheorem{definition}{Definition}[section]
\newtheorem{theorem}{Theorem}[section]
\newtheorem{lemma}{Lemma}[section]
\newtheorem{remark}{Remark}[section]
\newtheorem{corollary}{Corollary}[section]
\newtheorem{observation}{Observation}[section]

\newtheorem{problem}{Problem}[section]

\usepackage{fancyhdr}
\fancypagestyle{plain}{ %
	\fancyhf{} %
	\fancyfoot[C]{\thepage}

      }

\usepackage{titling}
\usepackage[hang,flushmargin]{footmisc}
\usepackage[yyyymmdd]{datetime}

\newcommand\nnfootnote[1]{%
  \begin{NoHyper}
    \renewcommand\thefootnote{}\footnote{#1}%
    \addtocounter{footnote}{-1}%
  \end{NoHyper}
}

\def\1{\mathbb 1}

\DeclareMathOperator{\poly}{poly}

\begin{document}

\title{Temporal matching in trees}

\author{
  M\'ark Hunor Juh\'asz\thanks{Department of Operations Research, E\"otv\"os Lor\'and University Budapest, P\'azm\'any P.\ s.\ 1/c, Budapest H-1117. E-mail: {\tt markh.shepherd@gmail.com}}
  \and P\'eter Madarasi\thanks{HUN-REN Alfr\'ed R\'enyi Institute of Mathematics, and Department of Operations Research, E\"otv\"os Lor\'and University Budapest, P\'azm\'any P.\ s.\ 1/c, Budapest H-1117. E-mail: {\tt madarasi@renyi.hu}}
}

\date{\vspace*{-24pt}}

\maketitle
\nnfootnote{\vspace{6pt}Corresponding author: P\'eter Madarasi\\Supported by the NRDI EXCELLENCE-24 grant no.~151504 Combinatorics and Geometry.}

\begin{abstract}

  We study maximum matching problems in temporal graphs whose underlying graph is a tree.
  We consider two temporal models.
  In a $\Delta$-matching, selected time edges sharing an endpoint must have time ticks differing by at least $\Delta$.
  In a $\gamma$-matching, the selected objects are blocks of $\gamma$ consecutive appearances of the same underlying edge.
  We also consider the related ordered static problem of $d$-distance matchings.

  We show that maximum $\Delta$-matching remains NP-hard on temporal trees for every $\Delta\geq 2$, even in the sparse case where each edge appears at most twice.
  Using a reduction between the temporal models, we obtain the analogous result for maximum $\gamma$-matching on temporal trees, even when each edge admits at most two $\gamma$-edges.
  We also show, via a reduction from $d$-distance matching, that maximum $\gamma$-matching is APX-hard even when the underlying graph is bipartite.

  Complementing these hardness results, we identify several tractable cases.
  We prove that maximum $\Delta$-matching is polynomial-time solvable on temporal trees in which every edge appears exactly once, and that maximum $\gamma$-matching is polynomial-time solvable when each edge admits at most one $\gamma$-edge.
  We also give dynamic-programming algorithms under bounded local-use and local-sparsity assumptions, and derive polynomial-time solvability of maximum $d$-distance matching when the input bipartite graph is a tree.
  Finally, we prove that both maximum $\Delta$-matching and maximum $\gamma$-matching admit polynomial-time approximation schemes on temporal trees.
\medskip

\end{abstract}

\section{Introduction}

Temporal graphs have become a widely used framework for modeling systems whose connections change over time, and they have recently attracted attention from several directions~\cite{survey1,survey3,survey2,survey4}.
In this paper, we study extensions of the classical maximum matching problem to the temporal setting, with a special focus on trees.

In a static graph, a matching is a set of independent edges, that is, no two selected edges share an endpoint.
In a temporal graph, the underlying graph is fixed, and each edge is active only at specified discrete times.
Thus an edge may be active at one time tick, inactive at a later tick, and active again later still.
When an edge $e$ is active at time tick $\tau$, the pair $(e,\tau)$ is called a time edge.
There is no unique canonical extension of the classical notion of matching to the temporal setting, since the appropriate independence condition depends on the intended effect of selecting a time edge.
Should the endpoints of a selected time edge be unavailable for a short recovery period?
Should an edge have to remain active for several consecutive ticks in order to be selected?
Different definitions lead to genuinely distinct matching models.

This paper studies two temporal matching models and a closely related ordered static problem.
In a $\Delta$-matching, the selected objects are individual time edges.
Two selected time edges may share a vertex only if their time ticks differ by at least $\Delta$.
Thus $\Delta$ is a separation parameter: after a vertex is matched at some time tick, it cannot be matched again within the next $\Delta-1$ ticks.
When $\Delta=1$, this only forbids using a vertex twice at the same tick; for larger $\Delta$, the choices made at nearby ticks interact.

The second model is $\gamma$-matching.
Here a selected object is a block of $\gamma$ consecutive appearances of the same underlying edge.
For instance, if an edge $e$ is active at all ticks in $[\tau,\tau+\gamma-1]$, then these appearances form one possible $\gamma$-edge.
Two such objects are compatible if the corresponding underlying edges are disjoint, or if their time intervals are disjoint.
Thus the $\gamma$-model captures edge selections that last for a consecutive interval of time, whereas the $\Delta$-model captures instantaneous edge selections subject to a temporal separation rule at the endpoints.

A related static problem is $d$-distance matching, or $d$-matching for short.
It is defined on a bipartite graph $G=(S,T,E)$ with the side $S=\{s_1,\dots,s_n\}$ ordered.
The goal is to choose as many edges as possible so that each vertex of $S$ is incident with at most one chosen edge, and whenever two chosen edges share an endpoint in $T$, their endpoints in $S$ have indices differing by at least $d$.
Although this definition is not temporal, the ordering of $S$ behaves like a time axis.
In this paper, we make this connection explicit and use it in two ways.
First, a reduction to $\gamma$-matching transfers the known APX-hardness of $d$-matching to maximum $\gamma$-matching on temporal graphs whose underlying graph is bipartite.
Second, viewing $d$-matching as a $\Delta$-matching instance in which every edge appears exactly once yields polynomial-time solvability of $d$-matching when the input bipartite graph is a tree.

\subsection{Previous work}

The maximum $\Delta$-matching problem was introduced and studied in~\cite{Delta1} under the name temporal matching.
It was shown there that the problem is APX-complete even if $\Delta=2$, the lifetime is $3$, and every edge of the underlying graph appears only once.
The corresponding decision problem was also shown to be NP-complete under several strong restrictions, including the case where $\Delta=2$ and the underlying graph is a path.
On the positive side,~\cite{Delta1} gave a $\frac{\Delta}{2\Delta-1}$-approximation algorithm and an exact FPT algorithm parameterized by the combination of $\Delta$ and the size of a maximum matching in the underlying graph.
A faster FPT algorithm was later obtained in~\cite{Delta2}, parameterized by the so-called $\Delta$-vertex-cover number.

The maximum $\gamma$-matching problem was introduced in~\cite{Gamma1} under the name temporal matching.
It was shown there that the problem is NP-hard for every $\gamma>1$.
The same work also gave a kernelization algorithm parameterized by the solution size and a $\frac{1}{2}$-approximation algorithm.
The problem was further studied in~\cite{Gamma2}, where an exact FPT algorithm parameterized by the number of vertices was obtained for the general case.
The same paper also gave a PTAS for temporal geometric graphs of bounded density.
Both the $\gamma$-model and the $\Delta$-model have appeared under the name temporal matching in the literature, but we keep the two notions separate throughout the paper.

Other temporal matching models select underlying edges rather than individual time edges or fixed-length blocks.
In the maximum $0$-$1$ timed matching problem~\cite{MandalGupta2022}, two selected edges are incompatible if they share an endpoint and overlap in time.
It was shown that the problem is NP-complete on rooted temporal trees when each edge is associated with at most two time intervals, and polynomial-time solvable on rooted temporal trees when each edge is associated with a single time interval.
A related edge-level temporal matching model, together with temporal edge cover, was studied in~\cite{CioniDondiMarinoSchoetersSilva2025}; NP-completeness results were obtained there, including for temporal graphs whose underlying graph is a tree.
These models differ from the $\gamma$- and $\Delta$-models studied here, since the selected objects are static edges rather than time edges or consecutive blocks of appearances.

The $d$-matching problem was introduced in~\cite{Madarasi1} under the name $d$-distance matching.
The problem was shown to be NP-complete in general.
The same paper also gave an FPT algorithm parameterized by $d$ and a polynomial-time algorithm for the case where $|T|$ is constant.
From an approximation and polyhedral point of view,~\cite{Madarasi1} proved that the integrality gap of the natural integer programming formulation is at most $2-\frac{1}{2d-1}$, and gave an LP-based approximation algorithm for the weighted problem with approximation ratio $\frac{2d-1}{4d-3}$.
It also gave a combinatorial $\frac{d}{2d-1}$-approximation algorithm for the weighted problem and a local-search $\left(\frac{2}{3}-\varepsilon\right)$-approximation algorithm for the unweighted problem, for every constant $0<\varepsilon<\frac{2}{3}$.
Further results were obtained in~\cite{Madarasi2}, which studies the more general $d$-distance $b$-matching problem and its cyclic variant.
Among other results,~\cite{Madarasi2} proved APX-hardness for the unweighted cardinality version, improved the approximation ratio for the ordinary non-cyclic $d$-matching problem to $\frac{d}{2(d-1)}$ for $d\geq 2$, and obtained tight or improved integrality-gap bounds for the corresponding natural integer programming formulations.

The preceding results show that all three problems are computationally difficult in general.
We study these problems under one of the most basic structural restrictions: the underlying graph is a tree, or, for $d$-matching, the input bipartite graph is a tree.
This restriction is natural because maximum matching in static trees is solvable by elementary dynamic programming.
Our hardness results show that, for temporal matchings, tree structure does not by itself overcome the difficulty caused by time-labels and temporal compatibility rules.
On the positive side, the same tree structure supports polynomial-time algorithms in restricted cases and approximation schemes in general.

\subsection{Our results}

Our first main result strengthens the known hardness of maximum $\Delta$-matching on temporal trees.
We prove that maximum $\Delta$-matching is NP-hard on temporal trees for every $\Delta\geq 2$, even if every edge appears at most twice.
This is a strong restriction on the time-label sets, since each static edge contributes at most two time edges.
Using a reduction between the two temporal models, we also prove that maximum $\gamma$-matching is NP-hard on temporal trees, even if each edge admits at most two $\gamma$-edges.
A new reduction from $d$-matching to $\gamma$-matching also implies that maximum $\gamma$-matching is APX-hard even when the underlying graph is bipartite.

We then identify additional restrictions under which temporal trees admit polynomial-time algorithms.
If every edge of the temporal tree appears exactly once, then maximum $\Delta$-matching can be solved in polynomial time.
The corresponding case for $\gamma$-matching, where each edge admits at most one $\gamma$-edge, is also polynomial-time solvable.
We further give a dynamic-programming algorithm for instances in which some optimal solution selects only a bounded number of time edges incident to each vertex.
Finally, using the connection between $d$-matchings and $\Delta$-matchings, we prove that maximum $d$-matching is polynomial-time solvable when the input bipartite graph is a tree.

The hardness results rule out exact polynomial-time algorithms in full generality, unless $\mathrm{P}=\mathrm{NP}$.
Nevertheless, we prove that maximum $\Delta$-matching on temporal trees admits a polynomial-time approximation scheme.
For every fixed $\varepsilon>0$, the algorithm returns a $(1-\varepsilon)$-approximate solution in polynomial time.
It partitions the time axis into periodically repeated windows separated by gaps of length $\Delta-1$, solves the problem exactly inside the windows, and uses a counting argument over all shifts of the window pattern.
The same approach, together with the reduction from $\gamma$-matching to $\Delta$-matching, gives a polynomial-time approximation scheme for maximum $\gamma$-matching on temporal trees.

The paper is organized as follows.
Section~\ref{sec:definitions-section} introduces temporal graphs and the matching models, records the reductions between the three problems, and derives the bipartite APX-hardness consequence for $\gamma$-matching.
Section~\ref{sec:multiple-edge-subsection} proves the hardness results.
Section~\ref{sec:single-edge-subsection} contains the polynomial-time solvable cases.
Section~\ref{sec:ptas-subsection} proves the approximation results.

\section{Temporal graphs and matching models}\label{sec:definitions-section}

All graphs considered in the paper are finite and simple.
For $k\in\mathbb{Z}_{\geq 0}$, let $[k]=\{1,\dots,k\}$, so $[0]=\emptyset$.
For integers $a,b$, let $[a,b]=\{z\in\mathbb{Z}:a\leq z\leq b\}$; in particular, $[a,b]=\emptyset$ if $b<a$.
Running times are stated under the usual unit-cost arithmetic convention.
We next introduce the temporal-graph notation used throughout the paper.
Our definition follows the foundational work of Kempe, Kleinberg, and Kumar~\cite{Temporal-def}.

\begin{definition}
A \textbf{temporal graph} $\mathcal{G}=(G,\lambda)$ consists of a finite static graph $G=(V,E)$ and a \textbf{time-labeling function} $\lambda$ that assigns a non-empty finite set $\lambda(e)\subseteq\mathbb{Z}_{>0}$ to each edge $e\in E$.
An edge $e\in E$ is \textbf{active} at time $\tau\in\mathbb{Z}_{>0}$ if $\tau\in\lambda(e)$.
If $e$ is active at time $\tau$, then the pair $(e,\tau)$ is called a \textbf{time edge}, and we say that $e$ \textbf{appears} at time $\tau$.
If $E\neq\emptyset$, the \textbf{lifetime} $\mathcal{T}(\mathcal{G})<\infty$ is the largest time-label assigned to an edge, that is,
$\mathcal{T}(\mathcal{G})=\max\bigcup_{e\in E}\lambda(e)$.
If $E=\emptyset$, we set $\mathcal{T}(\mathcal{G})=0$.
The elements of $[\mathcal{T}(\mathcal{G})]$ are called \textbf{time ticks}; in particular, if $E=\emptyset$, then there are no time ticks.
Some time ticks may have no active edge.
\end{definition}

\subsection{\texorpdfstring{$\Delta$-matchings}{Delta-matchings}}

The $\Delta$-model selects individual time edges, with a separation constraint at common endpoints.

\begin{definition}
Two time edges $(e,\tau)$ and $(e',\tau')$ are \textbf{$\boldsymbol{\Delta}$-independent} if the edges $e,e'$ do not share an endpoint or their time ticks differ by at least $\Delta$, that is, $|\tau-\tau'|\geq\Delta$.
A \textbf{$\boldsymbol{\Delta}$-matching} of a temporal graph $\mathcal{G}$ is a set of pairwise $\Delta$-independent time edges of $\mathcal{G}$.
Equivalently, a set $M$ of time edges is a $\Delta$-matching if, in every interval of $\Delta$ consecutive ticks, the time edges of $M$ appearing in that interval form a matching in the underlying graph.
\end{definition}

Thus a vertex cannot be matched more than once within any $\Delta$ consecutive ticks.
Equivalently, one may view $\Delta$ as a recovery period after a vertex has been matched.

\begin{problem}
In the \textbf{maximum $\boldsymbol{\Delta}$-matching problem}, the input is a temporal graph $\mathcal{G}$ and a positive integer $\Delta$.
The goal is to find a $\Delta$-matching in $\mathcal{G}$ with the largest possible cardinality.
\end{problem}
\subsection{\texorpdfstring{$\gamma$-matchings}{gamma-matchings}}

The $\gamma$-model is the consecutive-use counterpart of the $\Delta$-model.
Here the basic object is not a single time edge, but a block of $\gamma$ consecutive appearances of one underlying edge.
Thus the model captures situations in which a match can be made only if the edge remains active throughout a fixed-length interval.

\begin{definition}
Let $\mathcal{G}=(G,\lambda)$ be a temporal graph with $G=(V,E)$.
For an edge $e\in E$ and a time tick $\tau\in\mathbb{Z}_{>0}$ such that
$[\tau,\tau+\gamma-1]\subseteq\lambda(e)$, define
$
  (e,\tau)_\gamma=\{(e,\tau') : \tau'\in[\tau,\tau+\gamma-1]\}.
$
We call $(e,\tau)_\gamma$ the \textbf{$\boldsymbol{\gamma}$-edge} of $e$ starting at time $\tau$.
Two $\gamma$-edges $(e,\tau)_\gamma$ and $(e',\tau')_\gamma$ are \textbf{$\boldsymbol{\gamma}$-independent} if the edges $e,e'$ do not share an endpoint or the intervals $[\tau,\tau+\gamma-1]$ and $[\tau',\tau'+\gamma-1]$ are disjoint.
A \textbf{$\boldsymbol{\gamma}$-matching} of $\mathcal{G}$ is a set of pairwise $\gamma$-independent $\gamma$-edges.
\end{definition}

\begin{problem}
In the \textbf{maximum $\boldsymbol{\gamma}$-matching problem}, the input is a temporal graph $\mathcal{G}$ and a positive integer $\gamma$.
The goal is to find a $\gamma$-matching in $\mathcal{G}$ with the largest possible cardinality.
\end{problem}

\subsection{\texorpdfstring{$d$-matchings}{d-matchings}}

The $d$-matching problem is static rather than temporal, but the ordered side of the input graph will play the role of a time axis in the reductions below.

\begin{problem}
In the \textbf{maximum $\boldsymbol{d}$-matching problem}, the input is a bipartite graph $G=(S,T,E)$ with $S=\{s_1,\dots,s_n\}$ and a positive integer $d$.
The goal is to find a subset $M\subseteq E$ of maximum cardinality such that every vertex in $S$ is incident with at most one edge of $M$, and whenever $s_it,s_jt\in M$ for some $t\in T$ with $i\neq j$, we have $|i-j|\geq d$.
A set of edges satisfying these two conditions is called a \textbf{$\boldsymbol{d}$-matching}.
\end{problem}

When we speak about NP-hardness for these maximization problems, we mean the corresponding decision versions, asking whether there is a feasible solution of cardinality at least a prescribed integer.
For approximation algorithms, we use the maximization convention that a $\varrho$-approximation, where $0<\varrho\leq 1$, returns a feasible solution of value at least $\varrho$ times the optimum value.

\subsection{Relations between the models}\label{sec:relations-subsection}

We record several reductions between the three models.
Some of them will be used later, while the reduction from $d$-matching to $\gamma$-matching also gives an immediate hardness consequence for $\gamma$-matching.
The first reduction translates maximum $\gamma$-matching into maximum $\Delta$-matching.

\begin{observation}\label{obs:gamma-to-delta}
Maximum $\gamma$-matching reduces in polynomial time to maximum $\Delta$-matching with ${\Delta=\gamma}$.
The reduction preserves optimum value.
\end{observation}

\begin{proof}
Consider an instance of the maximum $\gamma$-matching problem.
If no edge admits a $\gamma$-edge, then the empty matching is optimal.
Otherwise, keep only those edges that admit at least one $\gamma$-edge.
For each remaining edge $e$ and each interval $[\tau,\tau+\gamma-1]\subseteq\lambda(e)$, introduce one time edge of $e$ at time tick $\tau$, and set $\Delta=\gamma$.
Selecting this introduced time edge is equivalent to selecting the $\gamma$-edge $(e,\tau)_\gamma$.
Moreover, two $\gamma$-edges conflict exactly when the corresponding introduced time edges conflict in the $\Delta$-matching instance.
Thus feasible solutions correspond one-to-one and cardinalities are preserved.
\end{proof}

We will later use the following restricted reduction in the opposite direction.
The restriction ensures that, after each time edge is replaced by an interval of length $\Delta$, no additional $\gamma$-edges are created on the same underlying edge.

\begin{lemma}\label{lem:delta-to-gamma}
Consider the maximum $\Delta$-matching problem, where the underlying graph is $G=(V,E)$, the input temporal graph is $\mathcal{G}=(G,\lambda)$, and for every edge $e\in E$, $|x-y|>\Delta$ holds for every distinct $x,y\in\lambda(e)$.
In this case, the maximum $\Delta$-matching problem can be reduced to the maximum $\gamma$-matching problem.
\end{lemma}

\begin{proof}
Set $\gamma=\Delta$ and construct a temporal graph $\mathcal{G}'=(G,\lambda')$ as follows.
For every edge $e\in E$, write $\lambda(e)=\{\tau_1,\dots,\tau_k\}$, and set $\lambda'(e)=\bigcup_{i\in[k]}[\tau_i,\tau_i+\gamma-1]$.
Since distinct time-labels of $e$ in the original instance differ by more than $\Delta=\gamma$, the intervals $[\tau_i,\tau_i+\gamma-1]$ are pairwise disjoint and separated by at least one tick.
Therefore the only $\gamma$-edges of $e$ in $\mathcal{G}'$ are $(e,\tau_i)_\gamma$ for $i\in[k]$.
Selecting the time edge $(e,\tau_i)$ in the original instance is equivalent to selecting the $\gamma$-edge $(e,\tau_i)_\gamma$ in the new instance.
Because two intervals of length $\gamma$ are disjoint exactly when their starting ticks differ by at least $\gamma$, this correspondence preserves independence and hence feasibility.
It also preserves cardinality, giving the desired reduction.
\end{proof}

The following remark expresses the $d$-matching condition in a form parallel to the equivalent formulation of $\Delta$-matchings above.

\begin{remark}\label{rem:d-ekv-def}
In the bipartite graph $G=(S,T,E)$, a set $M\subseteq E$ is a $d$-matching if and only if, after restricting $M$ to any interval of at most $d$ consecutive vertices of $S$, the remaining edges form a matching.
Here, restricting $M$ to a set $X\subseteq S$ means keeping the edges of $M$ incident with vertices of $X$.
\end{remark}

This viewpoint gives a direct reduction from maximum $d$-matching to maximum $\gamma$-matching.

\begin{observation}\label{obs:d-to-gamma}
Maximum $d$-matching reduces in polynomial time to maximum $\gamma$-matching with $\gamma=d$.
The reduction preserves optimum value and produces a temporal graph whose underlying graph is bipartite.
\end{observation}

\begin{proof}
Let the input of the maximum $d$-matching problem be $G=(S,T,E)$, where $S=\{s_1,\dots,s_n\}$, and let $d\in\mathbb{Z}_{>0}$.
Construct a temporal graph $\mathcal{G}=(G',\lambda)$ with underlying graph $G'=(S\cup T,E)$.
For every edge $s_it\in E$, set $\lambda(s_it)=[i,i+d-1]$, and set $\gamma=d$.
Then every edge $s_it$ of the original instance gives exactly one $\gamma$-edge, namely $(s_it,i)_\gamma$.
Two chosen edges incident with the same vertex of $S$ give overlapping $\gamma$-edges, while two chosen edges $s_it$ and $s_jt$ incident with the same vertex of $T$ give independent $\gamma$-edges exactly when $|i-j|\geq d$.
Thus feasible solutions correspond one-to-one to $d$-matchings, and cardinalities are preserved.
\end{proof}

By Observation~\ref{obs:d-to-gamma} and the APX-hardness of the maximum $d$-matching problem~\cite{Madarasi2}, we obtain the following consequence.

\begin{corollary}
The maximum $\gamma$-matching problem is APX-hard, even if the underlying graph of the input temporal graph is bipartite.
\end{corollary}

Finally, we will use the following reduction from $d$-matching to the special case of $\Delta$-matching in which every edge appears exactly once.

\begin{lemma}\label{lem:d-delta-cor}
The maximum $d$-matching problem is a special case of the maximum $\Delta$-matching problem, in which every edge appears exactly once.
\end{lemma}

\begin{proof}
Let the input of the maximum $d$-matching problem be $G=(S,T,E)$, where $S=\{s_1,\dots,s_n\}$, and let $d\in\mathbb{Z}_{>0}$.
Construct a temporal graph $\mathcal{G}=(G,\lambda)$ on the same underlying graph by setting $\lambda(s_it)=\{i\}$ for every edge $s_it\in E$, and set $\Delta=d$.
Two selected edges incident with the same vertex of $T$ are $\Delta$-independent exactly when their indices in $S$ differ by at least $d$.
Two selected edges incident with the same vertex of $S$ are never compatible, because they appear at the same time tick.
Thus feasible $d$-matchings in $G$ are in one-to-one correspondence with $\Delta$-matchings in $\mathcal{G}$, and the correspondence preserves cardinality.
Every edge in the constructed temporal graph appears exactly once.
\end{proof}

\section{Hardness of the problems}\label{sec:multiple-edge-subsection}
The maximum $\Delta$-matching problem is known to be NP-hard on temporal trees~\cite{Delta1}.
We strengthen this result by showing that, for every $\Delta\geq 2$, the problem remains hard even if every edge appears at most twice.
We then derive the corresponding hardness result for the maximum $\gamma$-matching problem on trees, under the analogous restriction that each edge admits at most two $\gamma$-edges.

The reduction for $\Delta$-matching starts from the \emph{double matching problem} introduced in~\cite{Madarasi2}.
An instance of this problem is a bipartite graph $G=(S,T,E)$ together with two subsets $S_1,S_2\subseteq S$ such that $S_1\cup S_2=S$.
For each $i\in\{1,2\}$, let $E_i$ denote the set of edges between $S_i$ and $T$.
The goal is to find a maximum-cardinality set $M\subseteq E$ such that both $M\cap E_1$ and $M\cap E_2$ are matchings.

We need the following restricted form of double matching, in which every vertex in $S$ has degree exactly $2$.
This form follows from the proof of Theorem~3 in~\cite{Madarasi2}, because the instances constructed there already have this degree property.

\begin{lemma}\label{lem:double-matching}
The double matching problem is NP-hard to $\varrho$-approximate for every $\varrho>\frac{949}{950}$, even if every vertex $s\in S$ has degree exactly $2$.
Consequently, the double matching problem remains NP-hard under this restriction.
\end{lemma}

We now reduce this restricted double matching problem to the maximum $\Delta$-matching problem on temporal trees in which every edge appears at most twice.
\begin{theorem}\label{thm:Delta-m-on-trees}
For every $\Delta\geq 2$, the maximum $\Delta$-matching problem is NP-hard, even if the underlying graph is a tree and every edge appears at most twice.
\end{theorem}
\begin{proof}
Let $G=(S,T,E)$, $S_1$, and $S_2$ be an instance of the restricted double matching problem from Lemma~\ref{lem:double-matching}.
It is enough to transform such an instance, for an arbitrary fixed $\Delta\geq 2$, into an instance of maximum $\Delta$-matching on a temporal tree.

We first construct an auxiliary bipartite graph $G'$, which will determine the time-labels in the temporal graph $\mathcal{G}=(\widehat{G},\lambda)$ constructed later.
\begin{figure}
\centering
\begin{tikzpicture}[
  vertex/.style={circle, fill=black, inner sep=0pt, minimum size=5pt},
  blackedge/.style={thick}
]

\node[vertex, label=above:{$s_1$}] (ls1) at (0,2) {};
\node[vertex, label=above:{$s_2$}] (ls2) at (1.5,2) {};
\node[vertex, label=above:{$s_3$}] (ls3) at (3,2) {};

\node[vertex, label=below:{$t_1$}] (lt1) at (0.9,0) {};
\node[vertex, label=below:{$t_2$}] (lt2) at (2.1,0) {};

\draw[decorate,decoration={brace,mirror,amplitude=8pt}]
  ($(ls2)+(0.45,0.45)$) -- ($(ls1)+(-0.45,0.45)$)
  node[midway, yshift=0.8cm, draw=none, fill=none] {\large $S_1$};

\draw[decorate,decoration={brace,mirror,amplitude=8pt}]
  ($(ls3)+(0.45,0.85)$) -- ($(ls2)+(-0.45,0.85)$)
  node[midway, yshift=0.8cm, draw=none, fill=none] {\large $S_2$};

\draw[blackedge] (lt1) -- (ls1);
\draw[blackedge] (lt2) -- (ls1);
\draw[blackedge] (lt1) -- (ls2);
\draw[blackedge] (lt2) -- (ls2);
\draw[blackedge] (lt1) -- (ls3);
\draw[blackedge] (lt2) -- (ls3);

\node[vertex, label=above:{$s_1$}] (rs1) at (8.7,2) {};
\node[vertex, label=above:{$s_2$}] (rs2) at (10.2,2) {};
\node[vertex, label=above:{$s_3$}] (rs3) at (11.7,2) {};

\node[vertex, label=below:{$t_1^1$}]      (rt11) at (6.0,0) {};
\node[vertex, label=below:{$t_1^2$}]      (rt12) at (7.2,0) {};
\node[vertex, label=below:{$t_1^3$}]      (rt13) at (8.4,0) {};
\node[vertex, label=below:{$d_1^{3,1}$}]  (rd1)  at (9.6,0) {};
\node[vertex, label=below:{$t_2^1$}]      (rt21) at (10.8,0) {};
\node[vertex, label=below:{$t_2^2$}]      (rt22) at (12.0,0) {};
\node[vertex, label=below:{$t_2^3$}]      (rt23) at (13.2,0) {};
\node[vertex, label=below:{$d_2^{3,1}$}]  (rd2)  at (14.4,0) {};

\draw[decorate,decoration={brace,mirror,amplitude=8pt}]
  ($(rs2)+(0.45,0.45)$) -- ($(rs1)+(-0.45,0.45)$)
  node[midway, yshift=0.8cm, draw=none, fill=none] {\large $S_1$};

\draw[decorate,decoration={brace,mirror,amplitude=8pt}]
  ($(rs3)+(0.45,0.85)$) -- ($(rs2)+(-0.45,0.85)$)
  node[midway, yshift=0.8cm, draw=none, fill=none] {\large $S_2$};

\draw[blackedge] (rt11) -- (rs1);
\draw[blackedge] (rt21) -- (rs1);

\draw[blackedge] (rt12) -- (rs2);
\draw[blackedge] (rt22) -- (rs2);

\draw[blackedge] (rt13) -- (rs3);
\draw[blackedge] (rt23) -- (rs3);

\end{tikzpicture}
\caption{Example of the construction of $G'$ for $\Delta=2$ in the proof of Theorem~\ref{thm:Delta-m-on-trees}.
  The left graph is the input $G$ of the double matching problem, and the right graph is the constructed bipartite graph $G'$.}
\label{fig:elso}
\end{figure}
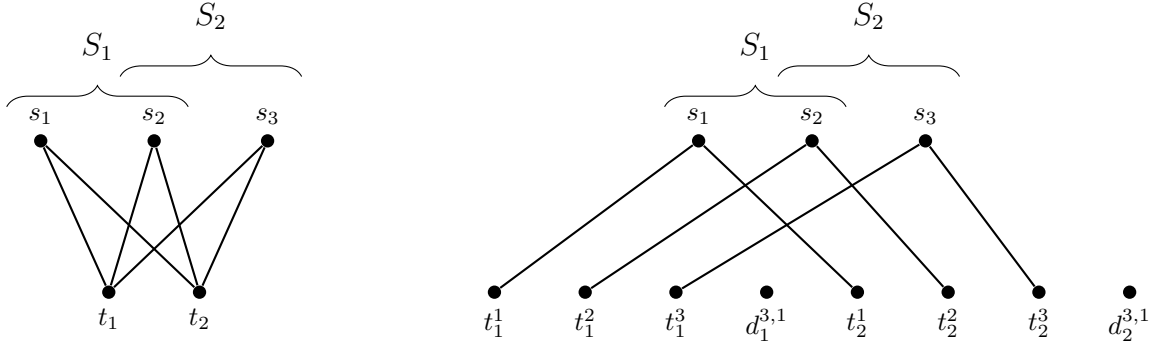
Fix orderings $(s_1,\dots,s_{|S|})$ of $S$ and $(t_1,\dots,t_{|T|})$ of $T$.
Let $G'=(S,T',E')$ be a bipartite graph, where $S$ is the same as in $G$.
For each $i\in[|T|]$, add to $T'$ the vertices $t_i^1,t_i^2,t_i^3$, the vertices $d_i^{a,j}$ for $a\in[2]$ and $j\in[\Delta-2]$, and the vertices $d_i^{3,j}$ for $j\in[\Delta-1]$.
Thus $T'$ contains exactly $(3(\Delta-1)+1)|T|$ vertices.
We order the vertices of $T'$ block by block, according to the ordering $t_1,\dots,t_{|T|}$ of $T$.
For each $i\in[|T|]$, the block associated with $t_i$ lists first $t_i^1$, then $d_i^{1,j}$ for $j\in[\Delta-2]$ in increasing order, then $t_i^2$, then $d_i^{2,j}$ for $j\in[\Delta-2]$ in increasing order, then $t_i^3$, and finally $d_i^{3,j}$ for $j\in[\Delta-1]$ in increasing order.
Let $(t'_1,\dots,t'_{(3(\Delta-1)+1)|T|})$ denote the resulting ordering of $T'$.
The edge set $E'$ is defined as follows.
For every edge $st_i\in E$, add exactly one edge to $E'$: add $st_i^1$ if $s\in S_1\setminus S_2$, add $st_i^2$ if $s\in S_1\cap S_2$, and add $st_i^3$ if $s\in S_2\setminus S_1$.
All vertices $d_i^{a,j}$ introduced above remain isolated.
An example of the construction of $G'$ is illustrated in Figure~\ref{fig:elso} for the case $\Delta=2$.

We now construct the temporal graph $\mathcal{G}=(\widehat{G},\lambda)$, whose underlying graph $\widehat{G}$ is a tree.
First, we define the \emph{core} of $\widehat{G}$ to be the star with center $v$ and leaf set $\{w_1,\dots,w_{|S|}\}$.
We assign time-labels to the edges of this star as follows.
For every $i\in[|S|]$, if $t'_x$ and $t'_y$ are the two neighbors of $s_i$ in $G'$, then set $\lambda(vw_i)=\{x,y\}$.
Since every vertex $s_i\in S$ has degree exactly $2$ in $G'$, this is well defined, and every edge $vw_i$ receives exactly two time-labels.

Before completing the construction of $\widehat{G}$, observe what the core represents.
Consider the maximum $\Delta$-matching problem restricted to the core star, with the additional constraint that at most one time tick may be selected on each edge of the core star.
Under the natural correspondence between an edge $s_it_j$ of $G$ and the time edge of $vw_i$ whose tick is the position of the corresponding copy of $t_j$ in $G'$, this restricted problem is exactly the double matching problem on $G$.
The ordering of $T'$ has been chosen so that two such core time edges conflict at the center $v$ exactly when the corresponding edges of $G$ cannot be chosen together in one of the two matchings $M\cap E_1$ or $M\cap E_2$.
We refer to this restricted problem on the core star as the \emph{core problem}.

We continue the construction of $\widehat{G}$ so that an optimal solution to the maximum $\Delta$-matching problem on the resulting temporal graph can be converted in polynomial time into an optimal solution to the core problem.
For each $i\in[|S|]$, we attach a subtree at $w_i$ as follows.
By possibly swapping the two labels, assume that $\lambda(vw_i)=\{x,y\}$ with $x<y$.
Let the neighbors of $s_i$ in the original graph $G$ be $t_p$ and $t_q$, with $p\neq q$.
By the construction of $G'$, the two neighbors of $s_i$ in $G'$ are $t_p^a$ and $t_q^a$ for the same $a\in[3]$, because the copy index is determined only by the membership of $s_i$ in $S_1$ and $S_2$.
The vertices of $T'$ are ordered block by block, and each block has size $3(\Delta-1)+1$.
Hence $y-x=|p-q|\bigl(3(\Delta-1)+1\bigr)$, and in particular $y-x>\Delta$.
Write $y-x=\ell(\Delta-1)+m$, where $\ell\in\mathbb{Z}_{\geq0}$ and $0\leq m<\Delta-1$.
Since $y-x>\Delta-1$, we have $\ell\in\mathbb{Z}_{>0}$.
We distinguish cases according to the parity of $\ell$ and the value of $m$.

If $\ell$ is even and $m\geq1$, then add new leaves $w_i^j$ for $j\in[\ell]$ adjacent to $w_i$.
For each $j\in[\ell]$, assign the single time-label $x+j(\Delta-1)$ to the edge $w_iw_i^j$, that is, set $\lambda(w_iw_i^j)=\{x+j(\Delta-1)\}$.

If $\ell$ is even and $m=0$, then we distinguish two cases.
First, if $\Delta=2$, then add an extra vertex $\widehat w_i$ and the edge $w_i\widehat w_i$.
Set $\lambda(w_i\widehat w_i)=\{x,y+1\}$.
Add new vertices $\widehat w_i^j$ for $j\in[\ell]$ and the edges $\widehat w_i\widehat w_i^j$ for $j\in[\ell]$.
Every new edge has exactly one time-label, namely $\lambda(\widehat w_i\widehat w_i^j)=\{x+j(\Delta-1)\}$ for every $j\in[\ell]$.
Second, if $\Delta \geq 3$, then add new vertices $w_i^j$ for $j\in[\ell]$ and the edges $w_iw_i^j$ for $j\in[\ell]$.
Every new edge has exactly one time-label, namely $\lambda(w_iw_i^j)=\{x+(j-1)(\Delta-1)+1\}$ for every $j\in[\ell]$.

If $\ell$ is odd and $m=0$, then add new vertices $w_i^j$ for $j\in[\ell-1]$ and the edges $w_iw_i^j$ for $j\in[\ell-1]$.
Every new edge has exactly one time-label, namely $\lambda(w_iw_i^j)=\{x+j(\Delta-1)\}$ for every $j\in[\ell-1]$.

If $\ell$ is odd and $m=1$, add an extra vertex $\widehat w_i$ and the edge $w_i\widehat w_i$.
Set $\lambda(w_i\widehat w_i)=\{x,y+1\}$.
Then add new vertices $\widehat w_i^j$ for $j\in[\ell+1]$ and the edges $\widehat w_i\widehat w_i^j$ for $j\in[\ell+1]$.
Every new edge has exactly one time-label, namely $\lambda(\widehat w_i\widehat w_i^j)=\{x+(j-1)(\Delta-1)+1\}$ for every $j\in[\ell+1]$.

Finally, if $\ell$ is odd and $m \geq 2$, add new vertices $w_i^j$ for $j\in[\ell+1]$ and the edges $w_iw_i^j$ for $j\in[\ell+1]$.
Every new edge has exactly one time-label, namely $\lambda(w_iw_i^j)=\{x+(j-1)(\Delta-1)+1\}$ for every $j\in[\ell+1]$.

\newcommand{\brokenedge}[3]{%
  \draw[blackedge] (#1) -- ($(#1)!0.5!(#2)$);
  \draw[blackedge] ($(#1)!0.5!(#2)$) -- (#2);
  \node[fill=white, inner sep=1pt] at ($(#1)!0.50!(#2)$) {$#3$};
}

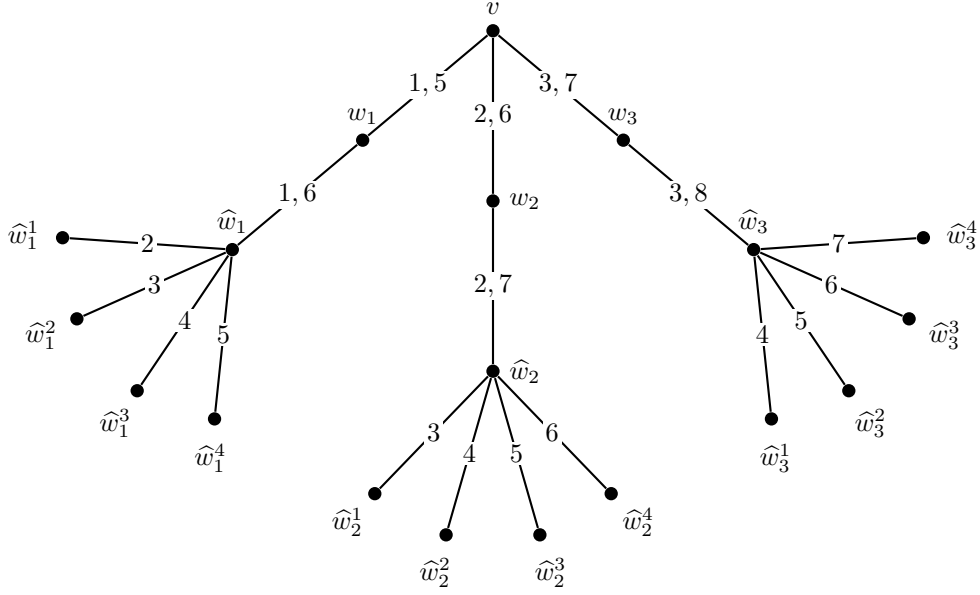
\begin{figure}
\centering
\begin{tikzpicture}[
  vnode/.style={circle, fill=black, inner sep=0pt, minimum size=5pt},
  blackedge/.style={thick}
]

\pgfmathsetmacro{\L}{2.25}     %

\pgfmathsetmacro{\aone}{-140}
\pgfmathsetmacro{\atwo}{-90}
\pgfmathsetmacro{\athree}{-40}

\pgfmathsetmacro{\bA}{-44}
\pgfmathsetmacro{\bB}{-16}
\pgfmathsetmacro{\bC}{16}
\pgfmathsetmacro{\bD}{44}

\node[vnode, label=above:{$v$}] (v) at (0,0) {};

\node[vnode, label=above:{$w_1$}] (w1) at ($(v)+(\aone:\L)$) {};
\node[vnode, label=right:{$w_2$}] (w2) at ($(v)+(\atwo:\L)$) {};
\node[vnode, label=above:{$w_3$}] (w3) at ($(v)+(\athree:\L)$) {};

\node[vnode, label=above:{$\widehat w_1$}] (wh1) at ($(v)!2!(w1)$) {};
\node[vnode, label=right:{$\widehat w_2$}] (wh2) at ($(v)!2!(w2)$) {};
\node[vnode, label=above:{$\widehat w_3$}] (wh3) at ($(v)!2!(w3)$) {};

\node[vnode] (wh11) at ($(wh1)+(\aone+\bA:\L)$) {};
\node[vnode] (wh12) at ($(wh1)+(\aone+\bB:\L)$) {};
\node[vnode] (wh13) at ($(wh1)+(\aone+\bC:\L)$) {};
\node[vnode] (wh14) at ($(wh1)+(\aone+\bD:\L)$) {};

\node at ($(wh1)!1.23!(wh11)$) {$\widehat w_1^1$};
\node at ($(wh1)!1.23!(wh12)$) {$\widehat w_1^2$};
\node at ($(wh1)!1.23!(wh13)$) {$\widehat w_1^3$};
\node at ($(wh1)!1.23!(wh14)$) {$\widehat w_1^4$};

\node[vnode] (wh21) at ($(wh2)+(\atwo+\bA:\L)$) {};
\node[vnode] (wh22) at ($(wh2)+(\atwo+\bB:\L)$) {};
\node[vnode] (wh23) at ($(wh2)+(\atwo+\bC:\L)$) {};
\node[vnode] (wh24) at ($(wh2)+(\atwo+\bD:\L)$) {};

\node at ($(wh2)!1.23!(wh21)$) {$\widehat w_2^1$};
\node at ($(wh2)!1.23!(wh22)$) {$\widehat w_2^2$};
\node at ($(wh2)!1.23!(wh23)$) {$\widehat w_2^3$};
\node at ($(wh2)!1.23!(wh24)$) {$\widehat w_2^4$};

\node[vnode] (wh31) at ($(wh3)+(\athree+\bA:\L)$) {};
\node[vnode] (wh32) at ($(wh3)+(\athree+\bB:\L)$) {};
\node[vnode] (wh33) at ($(wh3)+(\athree+\bC:\L)$) {};
\node[vnode] (wh34) at ($(wh3)+(\athree+\bD:\L)$) {};

\node at ($(wh3)!1.23!(wh31)$) {$\widehat w_3^1$};
\node at ($(wh3)!1.23!(wh32)$) {$\widehat w_3^2$};
\node at ($(wh3)!1.23!(wh33)$) {$\widehat w_3^3$};
\node at ($(wh3)!1.23!(wh34)$) {$\widehat w_3^4$};

\brokenedge{v}{w1}{1,5}
\brokenedge{v}{w2}{2,6}
\brokenedge{v}{w3}{3,7}

\brokenedge{w1}{wh1}{1,6}
\brokenedge{w2}{wh2}{2,7}
\brokenedge{w3}{wh3}{3,8}

\brokenedge{wh1}{wh11}{2}
\brokenedge{wh1}{wh12}{3}
\brokenedge{wh1}{wh13}{4}
\brokenedge{wh1}{wh14}{5}

\brokenedge{wh2}{wh21}{3}
\brokenedge{wh2}{wh22}{4}
\brokenedge{wh2}{wh23}{5}
\brokenedge{wh2}{wh24}{6}

\brokenedge{wh3}{wh31}{4}
\brokenedge{wh3}{wh32}{5}
\brokenedge{wh3}{wh33}{6}
\brokenedge{wh3}{wh34}{7}

\end{tikzpicture}
\caption{The tree $\widehat G$ obtained by applying the construction from the proof of Theorem~\ref{thm:Delta-m-on-trees} to the example shown in Figure~\ref{fig:elso}, for $\Delta=2$.}
\label{fig:masodik}
\end{figure}

This completes the construction of $\widehat{G}$, which is illustrated in Figure~\ref{fig:masodik}.
The graph $\widehat{G}$ is a tree, every edge has at most two time-labels, and the construction is polynomial in the size of the double matching instance.
We next prove that an optimal solution on $\widehat{G}$ yields an optimal solution to the core problem.
Call a solution of the maximum $\Delta$-matching problem on $\mathcal{G}$ \emph{normalized} if it contains at most one time edge from each edge of the core star.
The next lemma shows that the attached subtrees force no loss of optimality when we restrict attention to solutions using at most one appearance of each core edge.

\begin{lemma}\label{lem:normalized-lemma-1}
For every optimal solution of the maximum $\Delta$-matching problem on $\mathcal{G}$, there exists an optimal normalized solution of the same size, and such a solution can be obtained in polynomial time from it.
\end{lemma}
\begin{proof}
Let $M$ be an optimal solution, and suppose that $M$ contains both time edges $(vw_i,x)$ and $(vw_i,y)$ of some core edge $vw_i$, where $x<y$.
Let $T_i$ be the subtree attached at $w_i$, excluding the core edge $vw_i$, and let $Q$ be the set of time edges of $M$ lying in $T_i$.
It is enough to show that, after deleting $(vw_i,y)$, there is a feasible set $Q'$ of time edges in $T_i$ that is compatible with $(vw_i,x)$ and satisfies $|Q'|\geq |Q|+1$.
We then keep exactly $|Q|+1$ time edges from $Q'$ and replace $Q\cup\{(vw_i,y)\}$ by these time edges.
This preserves the cardinality of the solution.
Since $T_i$ meets the rest of the graph only in $w_i$, and compatibility with the remaining core time edge $(vw_i,x)$ is ensured, the replacement cannot create any conflict outside~$T_i$.

Put $d=\Delta-1$.
In each star considered below, the labels on the leaf edges form a consecutive arithmetic progression with difference $d$.
Two consecutive labels in such a progression conflict at the center of the star, while labels at distance at least two in the progression are compatible, since $2d\geq \Delta$.
Hence from $q$ consecutive leaf labels one can select exactly $\lceil q/2\rceil$ pairwise compatible time edges, by taking every other label.

First suppose that no vertex $\widehat w_i$ is introduced.
Then $T_i$ is a star centered at $w_i$.
In each of the relevant cases, the first leaf label is the only leaf label conflicting with $(vw_i,x)$, and the last leaf label is the only leaf label conflicting with $(vw_i,y)$.
This follows directly from the displayed label sequences: the next leaf label after the first is at distance at least $\Delta$ from $x$, and the leaf label before the last is at distance at least $\Delta$ from $y$.
Thus, while both core time edges are present, $Q$ can use only the leaf labels between the first and the last one, whereas after deleting $(vw_i,y)$ only the first leaf label is ruled out.
If $\ell$ is even, there are $\ell$ leaf labels, so $|Q|\leq \ell/2-1$, while after the deletion the remaining labels contain $\ell/2$ compatible choices.
If $\ell$ is odd and $m=0$, there are $\ell-1$ leaf labels, so $|Q|\leq (\ell-3)/2$, while after the deletion the remaining labels contain $(\ell-1)/2$ compatible choices.
If $\ell$ is odd and $m\geq2$, there are $\ell+1$ leaf labels, so $|Q|\leq (\ell-1)/2$, while after the deletion the remaining labels contain $(\ell+1)/2$ compatible choices.
Hence each case without an auxiliary vertex gives a feasible set $Q'$ with $|Q'|\geq |Q|+1$.

It remains to consider the cases in which $\widehat w_i$ is introduced.
The edge $w_i\widehat w_i$ has labels $x$ and $y+1$.
While both core time edges are present, neither of these can be selected: the first conflicts with $(vw_i,x)$, and the second conflicts with $(vw_i,y)$.
Thus $Q$ uses only leaf edges incident to $\widehat w_i$.
After deleting $(vw_i,y)$, the time edge $(w_i\widehat w_i,y+1)$ is compatible with $(vw_i,x)$, because $y+1-x>\Delta$.
Among the leaf labels incident to $\widehat w_i$, it conflicts only with the last one.
If $\Delta=2$, $\ell$ is even, and $m=0$, then the $\ell$ leaf labels incident to $\widehat w_i$ contribute at most $\ell/2$ time edges to $Q$, while the labels remaining after the last one is omitted contain $\ell/2$ compatible choices.
If $\ell$ is odd and $m=1$, then the $\ell+1$ leaf labels incident to $\widehat w_i$ contribute at most $(\ell+1)/2$ time edges to $Q$, while the labels remaining after the last one is omitted contain $(\ell+1)/2$ compatible choices.
Together with $(w_i\widehat w_i,y+1)$, this again gives a feasible set $Q'$ with $|Q'|\geq |Q|+1$.

Repeating this local replacement for every core edge on which two time edges are selected yields an optimal normalized solution.
Since each replacement is local and the attached subtrees have polynomial size, the normalization can be carried out in polynomial time.
\end{proof}

It remains to show that the restriction of an optimal normalized solution to the core star is optimal for the core problem.

\begin{lemma}\label{lem:normalized-lemma-2}
For every optimal normalized solution to the maximum $\Delta$-matching problem on $\mathcal{G}$, its restriction to the core star is an optimal solution to the core problem.
\end{lemma}
\begin{proof}
For every core edge $vw_i$, we use the following property of the subtree attached at $w_i$.
Among the time edges in this subtree, the maximum number that can be selected compatibly with the choice on $vw_i$ is the same whether no time edge of $vw_i$ is selected or exactly one time edge of $vw_i$ is selected.
If exactly one time edge of $vw_i$ is selected, this maximum is independent of which of the two time edges of $vw_i$ is selected.
We justify this property once, using the notation from the construction.
In the cases without an auxiliary vertex, the leaf labels form a consecutive arithmetic progression with difference $\Delta-1$ and with an even number $r$ of terms, where $r$ is $\ell$, $\ell-1$, or $\ell+1$ according to the case.
As in the proof of Lemma~\ref{lem:normalized-lemma-1}, the maximum number of selectable leaf labels is $r/2$.
If $(vw_i,x)$ is selected, exactly the first leaf label is ruled out, and if $(vw_i,y)$ is selected, exactly the last leaf label is ruled out; in either case the maximum remains $\lceil (r-1)/2\rceil=r/2$.
In the case where $\Delta=2$, $\ell$ is even, and $m=0$, the non-core labels form the integer interval $[x,y+1]$.
With no core time edge selected, this sequence has maximum $(\ell+2)/2$; selecting one core time edge excludes one endpoint and leaves $\ell+1$ consecutive labels, whose maximum is again $\ell/2+1$.
Finally, consider the case where $\ell$ is odd and $m=1$.
The labels $x$ and $y+1$ on $w_i\widehat w_i$ are compatible with all leaf labels except the first and the last one, respectively.
If no core time edge is selected, the two endpoint pairs contribute at most two labels in total, and the remaining $\ell-1$ leaf labels contribute at most $(\ell-1)/2$; this bound is attained by taking $x$, $y+1$, and every other remaining leaf label.
If $(vw_i,x)$ is selected, the pair formed by $y+1$ and the last leaf label contributes at most one, while the remaining $\ell$ leaf labels contribute at most $\lceil\ell/2\rceil=(\ell+1)/2$; this bound is attained by taking $y+1$ and every other leaf label except the last.
The case in which $(vw_i,y)$ is selected is symmetric.
Thus the property holds in every case.

Let $C$ denote the set of time edges of the core star.
Let $M$ be an optimal normalized solution whose restriction $M\cap C$ is not optimal for the core problem.
Then there exists a feasible core solution $N_c\subseteq C$ with $|N_c|>|M\cap C|$.
For each attached subtree, choose a maximum feasible set compatible with the core choice prescribed by $N_c$.
By the property above, the total number of chosen non-core time edges is at least $|M\setminus C|$.
Combining these choices with $N_c$ gives a feasible solution $M'$ on $\mathcal{G}$ with $M'\cap C=N_c$ and $|M'\setminus C|\geq |M\setminus C|$.
Hence $|M'|>|M|$, contradicting the optimality of $M$.
\end{proof}

Suppose that there were a polynomial-time algorithm for maximum $\Delta$-matching on the constructed temporal trees.
Applying it to $\mathcal{G}$, and then applying Lemma~\ref{lem:normalized-lemma-1}, we obtain in polynomial time an optimal normalized solution.
By Lemma~\ref{lem:normalized-lemma-2}, the restriction of this solution to the core star is an optimal solution to the core problem.
Under the correspondence described above, this gives an optimal solution to the original double matching instance.
Thus a polynomial-time algorithm for maximum $\Delta$-matching on the constructed trees would imply a polynomial-time algorithm for the restricted double matching problem, which is NP-hard by Lemma~\ref{lem:double-matching}.
This proves the theorem.
\end{proof}

Theorem~\ref{thm:Delta-m-on-trees} shows that, for every $\Delta\geq 2$, the maximum $\Delta$-matching problem is NP-hard on trees even if every edge appears at most twice.
We now derive the corresponding hardness result for maximum $\gamma$-matching.

\begin{theorem}\label{thm:Gamma-m-on-trees}
For every $\gamma\geq 2$, the maximum $\gamma$-matching problem is NP-hard, even if the underlying graph is a tree and each edge admits at most two $\gamma$-edges.
\end{theorem}
\begin{proof}
Fix $\gamma\geq 2$, and consider the temporal graph $\mathcal{G}$ constructed in the proof of Theorem~\ref{thm:Delta-m-on-trees} with $\Delta=\gamma$.
For every core edge $vw_i$, the two time-labels were chosen as the positions of $t_p^a$ and $t_q^a$ in the ordering of $T'$, where $p\neq q$ and the copy index $a\in[3]$ is the same for the two labels.
Since the vertices of $T'$ are ordered in blocks of size $3(\Delta-1)+1$, these two labels differ by a positive multiple of $3(\Delta-1)+1$, and hence by more than $\Delta$.
Every non-core edge has either one time-label or, in the two cases where the edge $w_i\widehat w_i$ is introduced, the two time-labels $x$ and $y+1$.
In the latter case, $y+1-x>\Delta$ by the preceding paragraph.
Thus any two distinct time-labels of the same edge in $\mathcal{G}$ differ by more than $\Delta$, so the condition of Lemma~\ref{lem:delta-to-gamma} holds for $\mathcal{G}$.
Applying that lemma with its $\gamma$-parameter equal to $\Delta$ reduces the maximum $\Delta$-matching instance on $\mathcal{G}$ to a maximum $\gamma$-matching instance.
This reduction does not change the underlying graph, and hence the underlying graph remains a tree.
Moreover, since every edge of $\mathcal{G}$ appears at most twice, each edge in the reduced instance admits at most two $\gamma$-edges.
Therefore, the maximum $\gamma$-matching problem is NP-hard under the stated restrictions.
\end{proof}

\section{Solvable cases}\label{sec:single-edge-subsection}

In Section~\ref{sec:multiple-edge-subsection}, we proved that the maximum $\Delta$-matching problem is NP-hard on trees even if every edge appears at most twice, and we established an analogous hardness result for the maximum $\gamma$-matching problem.
In this section, we consider the complementary setting: for $\Delta$-matching, every edge appears exactly once, and for $\gamma$-matching, each edge admits at most one $\gamma$-edge.
Under these restrictions, we show that both problems are polynomial-time solvable on trees.
We also present additional polynomial-time solvable cases and, finally, prove that the maximum $d$-matching problem is polynomial-time solvable on trees.

We first consider the single-appearance case, where each edge contributes only one time edge.
The resulting polynomial-time algorithm will later imply tractability of maximum $d$-matching on trees through the reduction from Lemma~\ref{lem:d-delta-cor}.
\begin{theorem}\label{thm:single-appear-theorem}
The maximum $\Delta$-matching problem can be solved in polynomial time for temporal graphs $\mathcal{G}=(G,\lambda)$ whose underlying graph $G=(V,E)$ is a tree and in which every edge appears exactly once.
\end{theorem}
\begin{proof}
  We solve the problem by dynamic programming.
  Choose an arbitrary vertex $r\in V$, and root the tree $G$ at $r$.
  For every vertex $v\neq r$, let $p(v)$ denote the parent of~$v$.
  Since every edge appears exactly once, each edge corresponds to a unique time edge; throughout this proof, we therefore speak simply of selecting edges.

  For every vertex $v\neq r$, let $\mathrm{OPT}_0(v)$ and $\mathrm{OPT}_1(v)$ denote the maximum sizes of $\Delta$-matchings in the subgraph induced by $p(v)$ and the vertices of the subtree rooted at $v$, under the condition that the edge $vp(v)$ is not selected and is selected, respectively.

  For each vertex $v\neq r$, we compute the values $\mathrm{OPT}_0(v)$ and $\mathrm{OPT}_1(v)$ bottom-up, starting from the leaves, as follows.
  Let $v'_1,\dots,v'_k$ be the children of $v$.
  For every $i\in [k]$, let $\tau_i$ denote the unique time tick of the edge $vv'_i$, that is, $\lambda(vv'_i)=\{\tau_i\}$.
  Choose an ordering of the children such that $\tau_1\leq \dots \leq \tau_k$.

  To compute $\mathrm{OPT}_0(v)$, for every $i\in[k]$, let $\alpha_i$ be the largest index $h<i$ for which $\tau_i-\tau_h\geq \Delta$; if no such index exists, then let $\alpha_i=0$.
  For every $i\in[0,k]$, let $A_i^{(0)}$ denote the maximum size of a $\Delta$-matching in the subgraph induced by $v$ and the vertices of the subtrees rooted at $v'_1,\dots,v'_i$, under the condition that the edge $vp(v)$ is not selected.
  Clearly, $A_0^{(0)}=0$.
  For every $i\in [k]$, we have
  \[
    A_i^{(0)}=\max\left\{
      A_{i-1}^{(0)} + \mathrm{OPT}_0(v'_i),\
      A_{\alpha_i}^{(0)} + \sum_{j=\alpha_i+1}^{i-1}\mathrm{OPT}_0(v'_j) + \mathrm{OPT}_1(v'_i)
    \right\}.
  \]
  Indeed, if the edge $vv'_i$ is not selected, then the subtree rooted at $v'_i$ contributes $\mathrm{OPT}_0(v'_i)$, while the first $i-1$ subtrees contribute $A_{i-1}^{(0)}$.
  If the edge $vv'_i$ is selected, then no edge $vv'_j$ with $j<i$ and $\tau_i-\tau_j<\Delta$ can be selected.
  By the definition of $\alpha_i$, this means that among the earlier edges, only $vv'_1,\dots,vv'_{\alpha_i}$ may be selected.
  Hence the first $\alpha_i$ subtrees contribute $A_{\alpha_i}^{(0)}$, the subtree rooted at $v'_i$ contributes $\mathrm{OPT}_1(v'_i)$, and each subtree rooted at $v'_j$ with $\alpha_i<j<i$ contributes $\mathrm{OPT}_0(v'_j)$.
  Therefore, $\mathrm{OPT}_0(v)=A_k^{(0)}$.

  To compute $\mathrm{OPT}_1(v)$, let $\tau_0$ be the unique time tick of the edge $vp(v)$, and let $i_1<\dots<i_\ell$ be precisely those indices $i\in [k]$ for which $|\tau_i-\tau_0|\geq \Delta$.
  Only the edges $vv'_{i_1},\dots,vv'_{i_\ell}$ may be selected together with $vp(v)$.
  Hence we apply the same recurrence as above to the restricted sequence of children $v'_{i_1},\dots,v'_{i_\ell}$.
  Let $A_\ell^{(1)}$ denote the value obtained by applying the above recurrence to the restricted sequence.
  Every remaining child $v'_i$ with $i\notin \{i_1,\dots,i_\ell\}$ cannot be matched to $v$, and therefore its subtree contributes $\mathrm{OPT}_0(v'_i)$.
  Adding also the contribution of the selected edge $vp(v)$, we obtain
  \[
    \mathrm{OPT}_1(v)
    =
    1+\sum_{i\in [k] \setminus \{i_1,\dots,i_\ell\}}\mathrm{OPT}_0(v'_i)+A_\ell^{(1)}.
  \]
  Finally, after all values for the children of $r$ have been computed, we apply the same recurrence used for $\mathrm{OPT}_0(v)$ to the root $r$, since the root has no parent edge.
  Let $\mathrm{OPT}(r)$ denote the value obtained in this way.
  Since every $\Delta$-matching of the original instance induces one of the choices considered at the root, and every solution counted by $\mathrm{OPT}(r)$ is feasible in the original instance, $\mathrm{OPT}(r)$ is the optimum value.
  The values are computed in polynomial time by processing the vertices bottom-up.
  This proves the theorem.
\end{proof}

The corresponding $\gamma$-matching statement follows by the reduction in Section~\ref{sec:relations-subsection}: if each edge admits at most one $\gamma$-edge, then the reduced $\Delta$-matching instance has each edge appearing once.

\begin{corollary}
  The maximum $\gamma$-matching problem can be solved in polynomial time if the underlying graph $G=(V,E)$ of the input temporal graph $\mathcal{G}=(G,\lambda)$ is a tree and, for every edge $e$, there is at most one $\gamma$-edge of $e$.
\end{corollary}
\begin{proof}
  Apply the reduction from maximum $\gamma$-matching to maximum $\Delta$-matching described in Section~\ref{sec:relations-subsection}.
  If no edge admits a $\gamma$-edge, then the empty matching is optimal.
  Otherwise, we construct the reduced instance only on those edges that admit a $\gamma$-edge.
  Since, by assumption, every edge admits at most one $\gamma$-edge, every edge in the reduced instance appears exactly once.
  The underlying graph of the reduced instance is a subgraph of the original tree, and hence it is a forest.
  Therefore, Theorem~\ref{thm:single-appear-theorem} applies to each non-trivial component of this forest, while isolated vertices contribute nothing.
  Since the reduction preserves feasible solutions and their cardinalities, this yields a polynomial-time algorithm for the original instance.
\end{proof}

The preceding two results impose per-edge restrictions.
We next allow arbitrary time-label sets, but assume that some optimal solution has bounded local use at every vertex.
The dynamic program below will later be used as the exact solver for the windows in the PTAS of Section~\ref{sec:ptas-subsection}.

\begin{theorem}\label{thm:PTAS-konstans}
  Consider the maximum $\Delta$-matching problem on a temporal graph $\mathcal{G}=(G,\lambda)$ whose underlying graph $G=(V,E)$ is a tree, and let $K\in\mathbb{Z}_{>0}$ be given.
  For every vertex $v\in V$, let $A_v=\bigcup_{e\ni v}\lambda(e)$ be the set of time ticks at which at least one edge incident to $v$ is active.
  Let $B=\max_{v\in V}|A_v|$.
  Suppose that there exists an optimal solution $M^*$ such that, for every vertex $v\in V$, the time edges of $M^*$ incident to $v$ use at most $K$ time ticks.
  Then an optimal solution can be found in $O\bigl((2eB)^K K^2 n\bigr)$ time, where $n=|V|$ and $e$ denotes Euler's number.
\end{theorem}
\begin{proof}
  We solve the problem by dynamic programming, using a rooted-tree setup similar to the one in the proof of Theorem~\ref{thm:single-appear-theorem}.
  Choose an arbitrary vertex $r\in V$, and root the tree $G$ at $r$.
  For every vertex $v\neq r$, let $p(v)$ denote the parent of~$v$.

  For every vertex $v\neq r$ and every subset $I\subseteq\lambda(vp(v))$ with $|I|\leq K$ and $|\tau-\sigma|\geq\Delta$ for every distinct $\tau,\sigma\in I$, let $\mathrm{OPT}(v,I)$ denote the maximum size of a $\Delta$-matching in the subgraph induced by $p(v)$ and the vertices of the subtree rooted at $v$, subject to the following two conditions: the edge $vp(v)$ is selected exactly at the time ticks in $I$, and at every vertex in the subtree rooted at $v$, the selected incident time edges use at most $K$ distinct time ticks.

  The values $\mathrm{OPT}(v,I)$ are computed bottom-up.
  Fix a vertex $v\neq r$, and suppose that the values for the children of $v$ have already been computed.
  Let $v'_1,\dots,v'_k$ be the children of $v$, and write $e_j=vv'_j$ for every $j\in[k]$.

  For this fixed vertex $v$, we compute all values $\mathrm{OPT}(v,I)$ simultaneously.
  We enumerate all sets $U\subseteq A_v$ of size at most $K$ such that $|\tau-\sigma|\geq\Delta$ for every distinct $\tau,\sigma\in U$.
  We then choose a subset $I\subseteq U$ with $I\subseteq\lambda(vp(v))$, representing the time ticks selected on the parent edge $vp(v)$.
  The remaining set $J=U\setminus I$ has to be assigned to the edges joining $v$ to its children.

  We consider every partition $\{L_1,\dots,L_h\}$ of $J$ into non-empty parts with $h\leq k$; if $J=\emptyset$, then we consider the empty partition.
  Each part $L_i$ is intended to be assigned to one child edge, while all remaining child edges are assigned the empty set.
  For a fixed partition, construct a bipartite graph $G'$ whose one side consists of $L_1,\dots,L_h$ and whose other side consists of $e_1,\dots,e_k$.
  We put an edge between $L_i$ and $e_j$ if and only if $L_i\subseteq\lambda(e_j)$, and assign weight $w'(L_i,e_j)=\mathrm{OPT}(v'_j,L_i)-\mathrm{OPT}(v'_j,\emptyset)$.
  We find a maximum-weight matching in $G'$ that covers the side $\{L_1,\dots,L_h\}$.
  If no such matching exists, then we discard the current partition.
  Otherwise, suppose that the matching assigns $L_i$ to $e_{\pi(i)}$.
  This gives the candidate value $|I|+\sum_{j=1}^{k}\mathrm{OPT}(v'_j,\emptyset)+\sum_{i=1}^{h}w'(L_i,e_{\pi(i)})$ for $\mathrm{OPT}(v,I)$.

  We now justify the recurrence.
  Every solution constructed in this way is feasible: the selected time ticks incident to $v$ are precisely the elements of $U$, and these are pairwise at distance at least~$\Delta$.
  Feasibility inside the child subtrees follows from the definition of the subproblems.

  Conversely, let $M$ be any feasible solution for the subproblem defining $\mathrm{OPT}(v,I)$.
  Let $U$ be the set of time ticks at which $M$ selects time edges incident to $v$.
  By the second condition in the definition of the subproblem, applied at the vertex $v$, we have $|U|\leq K$.
  Moreover, $U\subseteq A_v$, and since $M$ is a $\Delta$-matching, we have $|\tau-\sigma|\geq\Delta$ for every distinct $\tau,\sigma\in U$.
  The subset of $U$ used on the parent edge is exactly $I$, and the remaining set $U\setminus I$ is partitioned by the child edges used by $M$.
  Thus this choice of $U$, this subset $I$, and this partition are considered by the dynamic program.
  The assignment of the parts to the corresponding child edges gives a matching in $G'$ that covers $\{L_1,\dots,L_h\}$.
  For each child $v'_j$, the part of $M$ in the corresponding child subproblem contributes at most $\mathrm{OPT}(v'_j,L_i)$ if $e_j$ is assigned the set $L_i$, and at most $\mathrm{OPT}(v'_j,\emptyset)$ if $e_j$ is assigned the empty set.
  Therefore, the value of $M$ is at most one of the candidate values considered above, and the recurrence computes $\mathrm{OPT}(v,I)$.

  Finally, after all values for the children of $r$ have been computed, we apply the same enumeration-and-matching procedure to $r$, with the parent-edge set fixed to be empty.
  Thus, for the root, we take $I=\emptyset$ and $J=U$.
  Let $\mathrm{OPT}(r)$ denote the maximum value obtained in this way.

  We now prove that the dynamic program captures an optimal solution of the original instance.
  Let $M^*$ be an optimal solution satisfying the assumption of the theorem.
  For every vertex $v\in V$, the time ticks used by $M^*$ on edges incident to $v$ form a subset of $A_v$ of size at most $K$, and any two distinct such ticks differ by at least $\Delta$.
  Hence, when processing any non-root vertex $v$, the dynamic program considers the corresponding set $U$, the subset used on the parent edge, and the partition induced by the child edges.
  When processing the root $r$, it considers the corresponding set $U$ and the partition induced by the child edges.
  By induction from the leaves upward, the restriction of $M^*$ is represented in every subtree.
  Thus $\mathrm{OPT}(r)$ is at least the optimum value.
  Since every solution counted by $\mathrm{OPT}(r)$ is a feasible $\Delta$-matching for the original instance, equality follows.

  We now estimate the running time.
  Fix a vertex $v$, and let $k_v$ be the number of children of $v$.
  For fixed $m=|U|$, there are $\binom{|A_v|}{m}$ choices for $U$, at most $2^m$ choices for the subset assigned to the parent edge when $v\neq r$, and at most $m^m$ partitions of the remaining set, with the convention that the empty set has one partition.
  For the root, there is no parent-edge choice, so the same upper bound still applies.
  Hence the number of triples consisting of $U$, the subset assigned to the parent edge when applicable, and a partition of the remaining set is at most $\sum_{m=0}^{K}\binom{|A_v|}{m}2^m m^m$.
  Since $\binom{|A_v|}{m}\leq |A_v|^m/m!$ and $m^m/m!\leq e^m$ for $m\geq 1$, this sum is $O\bigl((2e|A_v|)^K\bigr)$.
  By the definition of $B$, we have $|A_v|\leq B$, and hence the number of triples is $O\bigl((2eB)^K\bigr)$.

  For each triple, we solve one maximum-weight bipartite matching problem in a graph with one side of size at most $K$ and the other side of size $k_v$.
  This takes $O(K^2\max\{k_v,1\})$ time.
  Thus the time spent at $v$ is $O\bigl((2eB)^K K^2\max\{k_v,1\}\bigr)$.
  Since $\sum_{v\in V}\max\{k_v,1\}=O(n)$, all values can be computed in $O\bigl((2eB)^K K^2 n\bigr)$ time.
  Storing the choices attaining the computed values yields an optimal solution within the same asymptotic running time.
\end{proof}

Using the reduction from $\gamma$-matching to $\Delta$-matching from Section~\ref{sec:relations-subsection}, the bounded-local-use result also yields the following analogue for $\gamma$-matchings.
\begin{corollary}
Consider the maximum $\gamma$-matching problem, where the underlying graph $G=(V,E)$ of the input temporal graph $\mathcal{G}=(G,\lambda)$ is a tree, and let $K\in\mathbb{Z}_{>0}$ be given.
Let
$
  B_\gamma=\max_{v\in V}|\{(e,\tau)_\gamma: e\in E,\ e\ni v,\ [\tau,\tau+\gamma-1]\subseteq\lambda(e)\}|.
$
If there exists an optimal solution $M^*$ for which the sets $M^*_v=\{(e,\tau)_\gamma\in M^* : e \text{ is incident to } v \}$ have size at most $K$ for every $v\in V$, then an optimal solution can be found in $O\bigl((2eB_\gamma)^K K^2 n\bigr)$ time, where $n=|V|$ and $e$ denotes Euler's number.
\end{corollary}

The hardness result of Section~3 shows that bounding the number of appearances of each edge is not enough for tractability: maximum $\Delta$-matching remains NP-hard on temporal trees even when every edge appears at most twice.
The next theorem gives a tractable case when the time-labels are locally sparse around every vertex.
It is enough that the total number of time-labels on the edges incident to each vertex is logarithmic in the input size.
\begin{theorem}
  For every fixed constant $c$, the maximum $\Delta$-matching problem can be solved in polynomial time on instances whose underlying graph $G=(V,E)$ is a tree and for which either $L\leq 1$ holds, or
  $
    \sum_{u\in N_G(v)}|\lambda(uv)| \leq c\log L
  $
  for every $v\in V$, where $L=\sum_{e\in E}|\lambda(e)|$ and $N_G(v)$ denotes the set of neighbors of $v$ in $G$.
\end{theorem}
\begin{proof}
  The case $L\leq 1$ is trivial, so we may assume that $L\geq 2$.
  We solve the problem by dynamic programming over a rooted tree.
  Choose an arbitrary vertex $r\in V$ and root $G$ at $r$.
  For every vertex $v\neq r$, let $p(v)$ denote the parent of $v$.

  For every vertex $v\neq r$ and every subset $I\subseteq\lambda(vp(v))$, let $\mathrm{OPT}(v,I)$ denote the maximum size of a $\Delta$-matching in the subgraph induced by $p(v)$ together with the vertices of the subtree rooted at $v$, subject to the condition that the edge $vp(v)$ is selected exactly at the time ticks in $I$.
  If no such $\Delta$-matching exists, then we set $\mathrm{OPT}(v,I)=-\infty$.
  The selected time edges on $vp(v)$ are counted in $\mathrm{OPT}(v,I)$.

  We compute these values bottom-up.
  Fix a vertex $v\neq r$ and a set $I\subseteq\lambda(vp(v))$, and suppose that all values for the children of $v$ have already been computed.
  If the time edges $(vp(v),\tau)$ with $\tau\in I$ are not pairwise $\Delta$-independent, then $\mathrm{OPT}(v,I)=-\infty$.
  Otherwise, let $v_1,\dots,v_k$ be the children of $v$, and write $e_i=vv_i$ for every $i\in[k]$.

  We enumerate all tuples $(L_1,\dots,L_k)$ such that $L_i\subseteq\lambda(e_i)$ for every $i\in[k]$.
  Such a tuple is feasible for $(v,I)$ if the time edges $(vp(v),\tau)$ with $\tau\in I$ and $(e_i,\sigma)$ with $i\in[k]$ and $\sigma\in L_i$ are pairwise $\Delta$-independent.
  Since all these time edges are incident with $v$, this is exactly the condition that their time ticks differ by at least $\Delta$ for every pair of distinct selected time edges.
  Each feasible tuple gives the candidate value $|I|+\sum_{i=1}^k\mathrm{OPT}(v_i,L_i)$, and candidates involving a value $-\infty$ are ignored.
  We define $\mathrm{OPT}(v,I)$ to be the maximum candidate value.

  The recurrence is correct for the following reason.
  Any candidate gives a feasible solution in the subproblem: feasibility inside each child subtree follows from the corresponding child value, and different child subgraphs meet each other and the parent edge only at $v$.
  Hence all possible conflicts between the pieces occur at $v$, and these are precisely the conflicts excluded by the feasibility condition on the tuple.
  Conversely, let $M$ be any feasible solution counted by $\mathrm{OPT}(v,I)$.
  For every child $v_i$, let $L_i$ be the set of time ticks at which $M$ selects the edge $vv_i$.
  Then $L_i\subseteq\lambda(e_i)$ for every $i$, and the tuple $(L_1,\dots,L_k)$ is feasible for $(v,I)$.
  Once the selected time ticks on $vv_i$ are fixed to be exactly $L_i$, the part of $M$ in the subproblem below $v_i$ contributes at most $\mathrm{OPT}(v_i,L_i)$.
  Therefore the size of $M$ is at most one of the candidate values considered above.
  Thus the recurrence computes the correct value.

  It remains to handle the root.
  Let $v_1,\dots,v_k$ be the children of $r$, and write $e_i=rv_i$ for every $i\in[k]$.
  We enumerate all tuples $(L_1,\dots,L_k)$ with $L_i\subseteq\lambda(e_i)$ for every $i$.
  Such a tuple is feasible at the root if the time edges $(e_i,\tau)$ with $i\in[k]$ and $\tau\in L_i$ are pairwise $\Delta$-independent.
  For every feasible tuple, we take the candidate value $\sum_{i=1}^k\mathrm{OPT}(v_i,L_i)$, again ignoring candidates involving $-\infty$.
  The maximum of these candidates is the optimum value for the original instance, by the same argument as above.

  We now estimate the running time.
  Put $b(v)=\sum_{u\in N_G(v)}|\lambda(uv)|$.
  By assumption, $b(v)\leq c\log L$ for every vertex $v$.
  Since every edge has a non-empty time-label set, every vertex has degree at most $b(v)$.

  For a fixed non-root vertex $v$, the number of possible states $I\subseteq\lambda(vp(v))$ is at most $2^{|\lambda(vp(v))|}$.
  For a fixed state $I$, the number of tuples $(L_1,\dots,L_k)$ that are enumerated is $\prod_{i=1}^k2^{|\lambda(e_i)|}$.
  Hence the total number of state--tuple pairs considered at $v$ is at most
  \[
    2^{|\lambda(vp(v))|}\prod_{i=1}^k2^{|\lambda(e_i)|}
    =
    2^{|\lambda(vp(v))|+\sum_{i=1}^k|\lambda(e_i)|}
    \leq
    2^{b(v)}
    \leq
    2^{c\log L}
    =
    L^{O(c)}.
  \]
  The same estimate applies to the root, except that there is no choice of a subset on a parent edge.

  For each state--tuple pair, feasibility can be checked by comparing all pairs of selected time edges incident with the currently processed vertex.
  There are at most $b(v)\leq c\log L$ such time edges.
  The corresponding candidate value is obtained by summing over the children of $v$, and the number of children is at most $b(v)$.
  Thus all values associated with a fixed vertex can be computed in $L^{O(c)}$ time.

  Since $G$ is a tree and every edge has a non-empty time-label set, we have $|E|\leq L$ and hence $|V|\leq L+1$ whenever $E\neq\emptyset$.
  Therefore the total running time is $L^{O(c)}$, which is polynomial in $L$ because $c$ is fixed.
  Storing the choices attaining the computed values yields an optimal solution within the same asymptotic running time.
  This proves the theorem.
\end{proof}

Applying the reduction from Observation~\ref{obs:gamma-to-delta}, admissible $\gamma$-edges become time edges of a $\Delta$-matching instance with $\Delta=\gamma$.
The local-sparsity condition is preserved under this translation, giving the following consequence.
\begin{corollary}
  For every fixed constant $c$, the maximum $\gamma$-matching problem can be solved in polynomial time on instances whose underlying graph $G=(V,E)$ is a tree and for which either $L_\gamma\leq 1$ holds, or
  $
    |\{(e,\tau)_\gamma : e\in E,\ e\ni v,\ [\tau,\tau+\gamma-1]\subseteq\lambda(e)\}| \leq c\log L_\gamma
  $
  for every $v\in V$, where $L_\gamma$ denotes the total number of $\gamma$-edges in the instance.
\end{corollary}

We close the section with the $d$-matching problem.
By Lemma~\ref{lem:d-delta-cor}, the single-appearance result for $\Delta$-matchings applies directly when the input graph of the $d$-matching instance is a tree.
This yields the following polynomial-time solvability result for $d$-matchings.

\begin{theorem}\label{thm:distance-tree-poly}
The maximum $d$-matching problem can be solved in polynomial time if the input bipartite graph $G=(S,T,E)$ is a tree.
\end{theorem}
\begin{proof}
  Let $S=\{s_1,\dots,s_n\}$ be the given ordering of the vertices in $S$.
  We use the reduction from Lemma~\ref{lem:d-delta-cor}, that is, we set $\Delta=d$ and construct a temporal graph $\mathcal{G}=(G,\lambda)$ on the same underlying graph by setting $\lambda(s_i t)=\{i\}$ for every edge $s_i t\in E$.
  Since $G$ is a tree, the underlying graph of $\mathcal{G}$ is a tree, and every edge appears exactly once.

  As observed in Lemma~\ref{lem:d-delta-cor}, feasible $d$-matchings in $G$ are in one-to-one correspondence with $\Delta$-matchings in $\mathcal{G}$, and this correspondence preserves cardinality.
  By Theorem~\ref{thm:single-appear-theorem}, a maximum $\Delta$-matching in $\mathcal{G}$ can be found in polynomial time.
  Translating such a solution back gives a maximum-cardinality $d$-matching in $G$.
\end{proof}

\section{Polynomial-time approximation schemes}\label{sec:ptas-subsection}

It is known that, on general temporal graphs, the maximum $\Delta$-matching problem is APX-hard already for $\Delta=2$ and $\mathcal{T}=3$~\cite{Delta1}. Hence, unless $\mathrm{P}=\mathrm{NP}$, no PTAS exists in that general setting.
In this section, we show that the problem admits a PTAS when the underlying graph of the input temporal graph is a tree.
The same approach also gives a PTAS for the maximum $\gamma$-matching problem on trees.
Although the algorithm could also be applied to the maximum $d$-matching problem, this is unnecessary here, since Theorem~\ref{thm:distance-tree-poly} solves that problem exactly on trees.

Fix an input temporal graph $\mathcal{G}=(G,\lambda)$ whose underlying graph $G=(V,E)$ is a tree, and let $\mathcal{T}$ be its lifetime.
We first dispose of the case $\Delta=1$.
If $\Delta=1$, then conflicts occur only between time edges sharing an endpoint at the same time tick.
Thus the problem can be solved exactly by computing a maximum matching independently in each active snapshot and taking the union of these matchings.
Hence, in the rest of the section, we assume that $\Delta\geq 2$.

The algorithm fixes a window length $k$, chosen as a function of $\Delta$ and $\varepsilon$, and considers shifted periodic choices of windows of length at most $k$ separated by $\Delta-1$ unused time ticks.
For each relevant choice, we solve the problem exactly inside every window and take the union of the obtained solutions.
The separation between consecutive windows ensures that this union is feasible.
The approximation guarantee follows from a counting argument over the full family of shifted choices, while the algorithm needs to inspect only a small representative subfamily.

We now introduce the required terminology, following~\cite{Delta1}.
Let $k$ be a positive integer with $\Delta\leq k\leq \mathcal{T}$, and set $p=k+\Delta-1$.
For each $\tau\in [\mathcal{T}-k+1]$, the \emph{$k$-window} $W_\tau$ is the interval $[\tau,\tau+k-1]$.
An interval of length at most $k-1$ that either begins at tick $1$ or ends at tick $\mathcal{T}$ is called a \emph{partial $k$-window}.

For an offset $a\in[0,p-1]$, declare a time tick $\tau\in[\mathcal{T}]$ to be \emph{covered} if the residue of $\tau-a$ modulo $p$ belongs to $[0,k-1]$.
The corresponding \emph{$k$-template} $\mathcal{S}_a$ is the family of maximal covered intervals after restricting to the lifetime $[\mathcal{T}]$.
Since $\Delta\geq 2$, we have $p>k$.
Thus every interval in a $k$-template has length at most $k$, intervals of length $k$ are $k$-windows, and shorter intervals can occur only at the endpoints and are partial $k$-windows.
Moreover, consecutive intervals of a template are separated by exactly $\Delta-1$ uncovered time ticks.
This definition gives one $k$-template for each of the $k+\Delta-1$ possible offsets; no maximum-cardinality condition is imposed.

For a $k$-template $\mathcal{S}$ and an interval $W\in\mathcal{S}$, let $M_W$ be a maximum $\Delta$-matching using only time edges whose ticks lie in $W$.
Let $M^{\mathcal{S}}=\bigcup_{W\in\mathcal{S}}M_W$.
This union is a feasible solution to the original problem because the intervals of $\mathcal{S}$ are disjoint and separated by $\Delta-1$ unused time ticks.
We say that a time tick is \emph{covered} by $\mathcal{S}$ if it belongs to an interval of $\mathcal{S}$.

The next lemma records the two properties of $k$-templates used in the approximation analysis.
It is a rephrasing of Lemma~24 in~\cite{Delta1} for our setting.

\begin{lemma}\label{lem:PTAS-lemma-1}
Let $k$ and $\mathcal{T}$ be integers such that $\Delta\leq k\leq \mathcal{T}$, with $\Delta\geq 2$.
Then there are exactly $k+\Delta-1$ different $k$-templates with respect to lifetime $\mathcal{T}$, and every time tick in $[\mathcal{T}]$ is covered by exactly $k$ of them.
\end{lemma}
\begin{proof}
  Put $p=k+\Delta-1$.
  There is one template for each offset $a\in[0,p-1]$.
  We first show that these templates are distinct.
  For an offset $a$, the uncovered residues modulo $p$ form a cyclic interval of length $\Delta-1$.
  Since $k\leq\mathcal{T}$, the template determines which of the ticks in $[k]$ are uncovered.
  A cyclic interval of length $\Delta-1$ in a cycle of length $p=k+\Delta-1$ is uniquely determined by its intersection with $k$ consecutive residues, because the complement of those residues is itself a cyclic interval of length $\Delta-1$.
  Hence two offsets giving the same template have the same uncovered residues modulo $p$, and therefore the same offset.
  Thus there are exactly $p=k+\Delta-1$ different templates.

  Now fix a time tick $\tau\in[\mathcal{T}]$.
  As $a$ ranges over $[0,p-1]$, the residue of $\tau-a$ modulo $p$ ranges over all residues modulo $p$.
  Exactly $k$ of these residues belong to $[0,k-1]$.
  Hence exactly $k$ templates cover~$\tau$.
\end{proof}

\begin{lemma}\label{lem:template-representatives}
  Let $\Delta\geq 2$ and let $k,\mathcal{T}$ be integers with $\Delta \leq k \leq \mathcal{T}$.
  Set $p=k+\Delta-1$, and let $A=\bigcup_{e\in E}\lambda(e)$ be the non-empty set of active time ticks.
  Then there is a set $R\subseteq[0,p-1]$ of offsets of size $O(|A|)$ such that for every offset $a\in[0,p-1]$, there exists $r\in R$ with the following property: an active time tick $\tau\in A$ is covered by $\mathcal{S}_a$ if and only if it is covered by~$\mathcal{S}_r$.
  Moreover, the set $R$ can be computed in polynomial time in $L=\sum_{e\in E}|\lambda(e)|$, and has size $O(L)$.
  Finally, any two offsets $a,r\in[0,p-1]$ that cover the same active time ticks satisfy $|M^{\mathcal{S}_a}|=|M^{\mathcal{S}_r}|$.
\end{lemma}
\begin{proof}
Fix an active time tick $\tau\in A$.
By the definition of the template $\mathcal{S}_a$, the tick $\tau$ is covered by $\mathcal{S}_a$ if and only if the residue of $\tau-a$ modulo $p$ belongs to $[0,k-1]$.
Equivalently, the offsets that cover $\tau$ are precisely the residues in $[\tau-k+1,\tau]$ modulo $p$.
Thus, as the offset varies cyclically over $[0,p-1]$, the covered or uncovered status of $\tau$ can change only when this cyclic block of covering offsets is entered or left.

For every $\tau\in A$, add the two residues $\tau-k+1 \pmod p$ and $\tau+1 \pmod p$ to a set $C$ of change points.
Let $C=\{c_1,\dots,c_s\}$, where $0\leq c_1<\dots<c_s\leq p-1$.
Consider the blocks of offsets $[c_i,c_{i+1}-1]$ for $i\in[s-1]$, together with the cyclic block $[c_s,p-1]\cup[0,c_1-1]$.
These blocks partition $[0,p-1]$.
For every active time tick $\tau\in A$, its covered or uncovered status is constant on each block, because all offsets at which this status can change are contained in $C$ and each block ends immediately before the next change point.
Choose one offset from each block and put it into $R$.
Now let $a\in[0,p-1]$ be arbitrary, and let $Q$ be the block containing $a$.
By construction, $R$ contains some offset $r\in Q$.
Since the covered or uncovered status of every active time tick is constant on $Q$, an active time tick $\tau\in A$ is covered by $\mathcal{S}_a$ if and only if it is covered by $\mathcal{S}_r$.
Thus $r$ is a representative for $a$.
Since there are exactly $|C|$ blocks, we have $|R|=|C|\leq 2|A|=O(|A|)$.
The construction only requires listing the active time ticks, computing residues modulo $p$, and sorting at most $2|A|$ residues.
Since $|A|\leq L$, this can be carried out in polynomial time in $L$.

It remains to prove the final assertion.
Suppose that $\mathcal{S}_a$ and $\mathcal{S}_r$ cover the same active time ticks.
Then they allow exactly the same time edges, because time edges exist only at active time ticks.
For any template $\mathcal{S}$, the matching $M^{\mathcal{S}}$ is a maximum $\Delta$-matching among all time edges whose time ticks are covered by $\mathcal{S}$.
Indeed, the time intervals of $\mathcal{S}$ are pairwise disjoint, and consecutive time intervals are separated by exactly $\Delta-1$ uncovered time ticks.
Hence any time edge in one time interval is $\Delta$-independent from any time edge in another time interval.
Therefore the optimum over the covered time edges is the sum of the optima inside the time intervals of $\mathcal{S}$, which is exactly the value of $M^{\mathcal{S}}$.
Since $\mathcal{S}_a$ and $\mathcal{S}_r$ allow the same set of time edges, they have the same optimum value over the covered time edges.
Consequently, $|M^{\mathcal{S}_a}|=|M^{\mathcal{S}_r}|$.
\end{proof}

For a fixed offset $r$, we compute $M^{\mathcal{S}_r}$ as follows.
It is enough to construct the intervals of $\mathcal{S}_r$ that contain at least one active time tick, since all other intervals contain no time edges and therefore contribute nothing.
Scan the active time ticks.
If an active time tick $\tau$ is covered by $\mathcal{S}_r$, and if $h$ is the residue of $\tau-r$ modulo $p$, then $h\in[0,k-1]$, and the untruncated covered time interval containing $\tau$ starts at $\tau-h$ and ends at $\tau-h+k-1$.
After restricting to the lifetime $[\mathcal{T}]$, this gives the interval $[\max\{1,\tau-h\},\min\{\mathcal{T},\tau-h+k-1\}]$.
Different active time ticks in the same interval of $\mathcal{S}_r$ may produce the same interval, so we keep only one copy of each interval obtained.
The resulting set is exactly the set of intervals of $\mathcal{S}_r$ that contain at least one active time tick.

For each such interval $W=[a,b]$, form a restricted temporal graph as follows.
Its underlying graph is the subgraph of $G$ consisting exactly of those edges $e$ for which $\lambda(e)\cap W\neq\emptyset$.
For each such edge $e$, replace its time-label set by $\{\tau-a+1: \tau\in\lambda(e)\cap W\}$.
Thus every edge in the restricted temporal graph has a non-empty time-label set.
This translation preserves all time differences inside $W$.
The resulting temporal graph has lifetime at most $k$, and its underlying graph is a forest.
Moreover, in any feasible $\Delta$-matching inside $W$, a fixed vertex can be incident to at most $\lceil k/\Delta\rceil$ selected time edges, because all selected time edges incident to that vertex must have pairwise time ticks at distance at least $\Delta$.
For the original instance, let $A_v=\bigcup_{e\ni v}\lambda(e)$, $B=\max_{v\in V}|A_v|$, and $n=|V|$.
For the restricted instance on $W$, the maximum number of active time ticks incident to a vertex is $B_W=\max_{v\in V}|A_v\cap W|$.
Since $A_v\cap W\subseteq A_v$ for every $v$, we have $B_W\leq B$.
Thus Theorem~\ref{thm:PTAS-konstans}, applied with $K=\lceil k/\Delta\rceil$, can be used on each component of the forest to compute $M_W$ in $O\bigl((2eB)^{\lceil k/\Delta\rceil}\lceil k/\Delta\rceil^2 n\bigr)$ time.
Taking the union of these matchings over all relevant intervals $W\in\mathcal{S}_r$ gives $M^{\mathcal{S}_r}$.

This gives the exact solution associated with a fixed representative offset.
We now combine this computation with Lemmas~\ref{lem:PTAS-lemma-1} and~\ref{lem:template-representatives}.

\begin{theorem}\label{thm:ptas-delta}
Let $\mathcal{G}=(G,\lambda)$ be a temporal graph whose underlying graph $G=(V,E)$ is a tree, and let $n=|V|$.
Write $L=\sum_{e\in E}|\lambda(e)|$ and $B=\max_{v\in V}|A_v|$, where $A_v=\bigcup_{e\ni v}\lambda(e)$.
For every $0<\varepsilon<1$, the maximum $\Delta$-matching problem on $\mathcal{G}$ admits a $(1-\varepsilon)$-approximation algorithm with running time
$
  O\bigl(L^2(2eB)^{\lceil 1/\varepsilon\rceil}\lceil 1/\varepsilon\rceil^2 n\bigr).
$
In particular, since $B\leq L$, the running time is $L^{O(1/\varepsilon)}\poly(n)$, and the problem admits a PTAS.
\end{theorem}
\begin{proof}
If $E=\emptyset$, then there are no time edges and the empty matching is optimal, so we are done.
Thus, since all time-label sets are non-empty, we have $L\geq 1$.
The case $\Delta=1$ is solved exactly by computing a maximum matching independently at each active time tick, as described above.
Since each snapshot is a forest and there are at most $L$ active time ticks, this can be done in polynomial time bounded by the displayed running time.
Assume, therefore, that $\Delta\geq 2$.
Let $A=\bigcup_{e\in E}\lambda(e)$ be the set of active time ticks.
Set $k=\max\{\Delta,\lceil(1-\varepsilon)(\Delta-1)/\varepsilon\rceil\}$, and let $p=k+\Delta-1$.
Then $k\leq \Delta/\varepsilon$, and hence $\lceil k/\Delta\rceil\leq \lceil 1/\varepsilon\rceil$.

If $\mathcal{T}<k$, then any feasible solution uses at most $\lceil \mathcal{T}/\Delta\rceil\leq \lceil k/\Delta\rceil$ time edges incident to any fixed vertex.
Thus Theorem~\ref{thm:PTAS-konstans}, applied with $K=\lceil k/\Delta\rceil$, solves the instance exactly in time $O\bigl((2eB)^{\lceil k/\Delta\rceil}\lceil k/\Delta\rceil^2 n\bigr)\leq O\bigl((2eB)^{\lceil 1/\varepsilon\rceil}\lceil 1/\varepsilon\rceil^2 n\bigr)$.
Since $L\geq 1$, this is bounded by the displayed running time.

We now assume that $k\leq\mathcal{T}$.
Compute the representative set $R$ from Lemma~\ref{lem:template-representatives}.
For every $r\in R$, compute $M^{\mathcal{S}_r}$ as described above, considering only intervals of $\mathcal{S}_r$ that contain at least one active time tick.
Let $M$ be a largest matching among these solutions.
The set $M$ is feasible, because it is equal to $M^{\mathcal{S}_r}$ for some offset $r$.
By Lemma~\ref{lem:template-representatives}, the algorithm considers only $O(L)$ representative offsets.
The set $R$ is constructed by sorting at most $2|A|\leq 2L$ residues.
For each representative offset, scanning the active time ticks produces at most $L$ relevant intervals, possibly with repetitions.
Thus the total number of interval descriptors produced over all representative offsets is $O(L^2)$.
For each such interval $W$, we construct the restricted instance by scanning the input time edges, keeping exactly those edges with at least one time tick in $W$, and translating the kept time ticks as described above.
This costs $O(L)$ time per interval.
Because the underlying graph is a tree, every edge $e=uv$ satisfies $|\lambda(e)|\leq |A_u|\leq B$, and hence $L\leq (n-1)B$.
Together with $\lceil 1/\varepsilon\rceil\geq 1$ and $B\geq 1$, this implies that the $O(L)$ construction cost is bounded by $O\bigl((2eB)^{\lceil 1/\varepsilon\rceil}\lceil 1/\varepsilon\rceil^2 n\bigr)$.

For each interval $W$, every feasible $\Delta$-matching inside $W$ selects at most $\lceil k/\Delta\rceil\leq \lceil 1/\varepsilon\rceil$ time edges incident to any fixed vertex.
For the restricted instance on $W$, the maximum number of active time ticks incident to a vertex is $B_W=\max_{v\in V}|A_v\cap W|$, and $B_W\leq B$.
By Theorem~\ref{thm:PTAS-konstans}, applied componentwise to the forest obtained from $W$, the subproblem for $W$ can be solved in $O\bigl((2eB)^{\lceil 1/\varepsilon\rceil}\lceil 1/\varepsilon\rceil^2 n\bigr)$ time.
There are $O(L^2)$ interval subproblems in total.
Hence the construction of $M$ takes $O\bigl(L^2(2eB)^{\lceil 1/\varepsilon\rceil}\lceil 1/\varepsilon\rceil^2 n\bigr)$ time.
The $O(L^2)$ time needed to construct the representative offsets and the interval descriptors is absorbed in this bound.
This gives the running time stated in the theorem.

\medskip
It remains to prove that $M$ has the desired size.
Let $M^*$ be an optimal solution.
Let $\mathcal{C}$ be the family of all $k$-templates.
By Lemma~\ref{lem:template-representatives}, for every offset $a\in[0,p-1]$ there exists $r\in R$ such that $|M^{\mathcal{S}_a}|=|M^{\mathcal{S}_r}|$.
Hence $|M|\geq |M^{\mathcal{S}}|$ for every $\mathcal{S}\in\mathcal{C}$.
For any set $X$ of time edges, let $X_\tau$ denote the elements of $X$ at time tick $\tau$.
For every $\mathcal{S}\in\mathcal{C}$, the intervals of $\mathcal{S}$ are disjoint, so $|M^{\mathcal{S}}|=\sum_{W\in\mathcal{S}}|M_W|$.
For every interval $W\in\mathcal{S}$, the time edges of $M^*$ whose ticks lie in $W$ form a feasible solution to the subproblem defining $M_W$, and therefore $|M_W|\geq \sum_{\tau\in W}|M^*_\tau|$.
Using Lemma~\ref{lem:PTAS-lemma-1}, we get
\[
(k+\Delta-1)|M|
\geq \sum_{\mathcal{S}\in\mathcal{C}}|M^{\mathcal{S}}|
\geq \sum_{\mathcal{S}\in\mathcal{C}}\sum_{W\in\mathcal{S}}\sum_{\tau\in W}|M^*_\tau|
= k|M^*|.
\]
Therefore $|M|\geq k|M^*|/(k+\Delta-1)$.
By the choice of $k$, we have $(\Delta-1)/(k+\Delta-1)\leq \varepsilon$, and hence $|M|\geq (1-\varepsilon)|M^*|$, thereby proving the desired approximation guarantee.
\end{proof}

The PTAS for $\Delta$-matching also yields the corresponding approximation scheme for $\gamma$-matching, via the reduction from Section~\ref{sec:relations-subsection}.

\begin{corollary}\label{cor:gamma-ptas}
Let $\mathcal{G}=(G,\lambda)$ be a temporal graph whose underlying graph $G=(V,E)$ is a tree, and let $n=|V|$.
Write $L=\sum_{e\in E}|\lambda(e)|$.
For every $0<\varepsilon<1$ and every $\gamma\in\mathbb{Z}_{>0}$, the maximum $\gamma$-matching problem on $\mathcal{G}$ admits a $(1-\varepsilon)$-approximation algorithm with running time $L^{O(1/\varepsilon)}\poly(n)$.
In particular, it admits a PTAS.
\end{corollary}
\begin{proof}
  Apply the reduction from maximum $\gamma$-matching to maximum $\Delta$-matching with $\Delta=\gamma$ from Observation~\ref{obs:gamma-to-delta}.
  If the reduced instance has no time edge, then the original instance contains no $\gamma$-edge, and the empty matching is optimal.
  Otherwise, the reduced instance has as its underlying graph a subgraph of $G$, and hence a forest.
  Moreover, it has at most $L$ time edges, since each maximal interval of $m$ consecutive time-labels on an edge contributes at most $m$ starting ticks of $\gamma$-edges, and hence at most $m$ introduced time edges, in the reduced instance.
  By Observation~\ref{obs:gamma-to-delta}, the reduction preserves feasible solutions and their cardinalities, and is therefore approximation-preserving.

  We apply Theorem~\ref{thm:ptas-delta} to each component of the reduced forest that contains at least one edge.
  Since the components are vertex-disjoint, the optimum value is the sum of the optimum values over the components, and the union of the componentwise $(1-\varepsilon)$-approximate solutions is a $(1-\varepsilon)$-approximate solution for the whole reduced instance.
  Translating this solution back through the reduction gives a $(1-\varepsilon)$-approximation for the original maximum $\gamma$-matching instance.
  Since the reduced instance has at most $L$ time edges in total, the total running time is $L^{O(1/\varepsilon)}\poly(n)$.
\end{proof}

\bibliographystyle{plain}
\bibliography{references}

\end{document}